\documentclass[a4paper,UKenglish]{lipics}

\usepackage{amssymb,amsmath}
\usepackage{stmaryrd}

\usepackage[backgroundcolor=orange!50, textsize=small]{todonotes}

\usepackage{thm-restate}
\usepackage[capitalise,noabbrev,nameinlink]{cleveref}
\usepackage{thm-restate}

\makeatletter
\let\c@theorem\@undefined
\let\theorem\@undefined
\let\endtheorem\@undefined
\let\lemma\@undefined
\let\endlemma\@undefined
\let\corollary\@undefined
\let\endcorollary\@undefined
\let\definition\@undefined
\let\enddefinition\@undefined
\let\example\@undefined
\let\endexample\@undefined
\let\remark\@undefined
\let\endremark\@undefined
\makeatother
\theoremstyle{plain}
\newtheorem{theorem}{Theorem}
\newtheorem{lemma}[theorem]{Lemma}
\newtheorem{corollary}[theorem]{Corollary}
\theoremstyle{definition}
\newtheorem{definition}[theorem]{Definition}
\newtheorem{example}[theorem]{Example}
\theoremstyle{remark}
\newtheorem{remark}[theorem]{Remark}
\crefname{lemma}{Lemma}{Lemmas}
\crefname{definition}{Definition}{Definitions}
\theoremstyle{plain}
\newtheorem{proposition}[theorem]{Proposition}
\crefname{proposition}{Proposition}{Propositions}
\theoremstyle{remark}
{\itshape}{\rmfamily}
\crefname{fact}{Fact}{Facts}

\crefname{claim}{Claim}{Claims}

\usepackage{tikz}
\usetikzlibrary{arrows,automata,backgrounds,decorations}
\usetikzlibrary{shapes.multipart}
\usepgflibrary{decorations.pathreplacing} 
\usetikzlibrary{arrows,calc,topaths}
\usepackage{tikz-3dplot}
\tikzstyle{small}=[font=\footnotesize]
\tikzset{
    every picture/.style={>=stealth,auto,node distance=2cm},
}
 
\usepackage{macros}
\usepackage{databoxes}

\usepackage[ruled]{algorithm}
\usepackage{algorithmic}
\usepackage{amsmath}
\floatname{algorithm}{Procedure}

\title{Linear Combinations of Unordered Data Vectors}

\author[1]{Piotr Hofman\footnote{Supported by Labex Digicosme, Univ.~Paris-Saclay, project VERICONISS and by Polish NSC grant 2013/09/B/ST6/01575.}}
\affil[1]{LSV, CNRS \& ENS Cachan, France}

\author[2]{J\'{e}r\^{o}me~Leroux}
\affil[2]
{LaBRI, CNRS, France}

\author[3]{Patrick Totzke\thanks{Supported by the EPSRC, grant EP/M027651/1.}}
\affil[3]{LFCS, University of Edinburgh, UK} 

\authorrunning{P.\ Hofman, J.\ Leroux, and P.\ Totzke}
\Copyright{Piotr Hofman, J\'{e}r\^{o}me~Leroux and Patrick Totzke}


\serieslogo{}
\volumeinfo
  {Billy Editor and Bill Editors}
  {2}
  {Conference title on which this volume is based on}
  {1}
  {1}
  {1}
\EventShortName{}
\DOI{10.4230/LIPIcs.xxx.yyy.p}

\begin{document}

\maketitle
\begin{abstract}
Data vectors generalise finite multisets: they
are finitely supported functions into a commutative monoid.
We study the question if a given data vector can be expressed
as a finite sum 
of others,
only assuming that 1) the domain is countable
and 2) the given set of base vectors is finite up to permutations of the domain.

Based on a succinct representation of the involved permutations
as integer linear constraints, we derive that
positive instances can be witnessed
in a bounded subset of the domain.

For data vectors over a group 
we moreover
study when a data vector is reversible, that is, if its inverse is
expressible using only nonnegative coefficients.
We show that
if all base vectors are reversible
then the expressibility problem reduces to checking membership
in finitely generated subgroups. 
Moreover, checking reversibility also reduces to such membership tests.

These questions naturally appear in the analysis of counter machines extended with unordered data:
namely, for data vectors over $(\Z^d,+)$ expressibility directly corresponds to
checking state equations for
Coloured Petri nets where tokens can only be tested for equality.
We derive that in this case, expressibility is in \NP, and in $\P$ for reversible instances.
These upper bounds are tight: they match the lower bounds for standard integer vectors (over singleton domains).

\end{abstract}

\section{Introduction}
Finite collections of named values are basic structures used in many areas of
theoretical computer science.
They can be used for instance to model databases snapshots or
define the operational semantics of programming languages.
We can formalize these as
functions $\vec v:\D\to X$
from some countable domain $\D$ of \emph{names} or \emph{data}, into some
value space $X$,
and call such functions ($X$-valued) \emph{data vectors}.
Often the actual names used are not relevant and instead one is
interested in data vectors \emph{up to renaming}, i.e., one wants to consider
vectors $\vec{v}$
and $\vec{w}$ equivalent if $\vec{v} = \vec{w}\circ\theta$ for some permutation
$\theta:\D\to\D$
of the domain.

We consider the case where the value space $X$ has additional algebraic structure.
Namely, we focus on data vectors
where the values are from some commutative monoid $(M,+,0)$
and where all but finitely many names are mapped to the neutral element.
A natural question then asks if a given data vector is expressible as
a sum  of vectors from a given set,
where the monoid operation is lifted to data vectors pointwise.

If the spanning set is finite only up to permutations,
this problem does not immediately boil down to solving finite system of linear equations. 
Also, one cannot simply lift operations on data vectors to equivalence
classes of data vectors
because the result of pointwise applying the operation
depends on the chosen representants.
For example, if we have data vectors mapping colours to integers,
then the vector $\datavec{-1/red!50,1/blue!50}$ ($red\mapsto -1, blue\mapsto 1$)
is equivalent to $\datavec{-1/yellow!50,1/blue!50}$. Yet still,
$$
    \datavec{1/red!50,-1/blue!50,2/yellow!50}
    \,
    +
    \,
    \datavec{-1/red!50,1/blue!50}
    \,
    =
    \,
    \datavec{2/yellow!50}
    \quad
    \neq
    \quad
    \datavec{1/red!50,1/yellow!50}
    \,
    =
    \,
    \datavec{1/red!50,-1/blue!50,2/yellow!50}
    \,
    +
    \datavec{-1/yellow!50,1/blue!50}
    .
$$

\noindent
We thus choose to keep permutations explicit
and consider the \emph{Expressibility} problem: 

\smallskip
\dproblem{
    A finite set $V$ of data vectors and a target vector $\vec{x}$.
    }{
    does $\vec{x}$ equal
    $\sum_{i=1}^{k}\vec{v_i}\circ\theta_i$ for some
    $\vec{v_i}\in V$ and permutations $\theta_i:\D\to\D$?
}
\smallskip

If the domain $\D$ is finite then this just asks if
some finite system of linear equations is satisfiable.
Over infinite domains this corresponds to find a solution of an infinite but regular set of linear equations.
For brevity, we will call
a vector $\vec{x}$ a \emph{permutation sum} of $V$
if it is expressible as 
sum of permutations of vectors in $V$ as above.

\subparagraph*{Contributions and Outline.}
We provide two reductions from the Expressibility problem to problems
of finding solutions for \emph{finite} linear systems over $(M,+)$.

The more general approach is presented in
\cref{sec:hist,sec:naturals} and ultimately works by bounding the number
of different data values necessary to express a \pproduct.
This is based on an analysis of objects we call \emph{histograms}, see \cref{sec:hist},
which sufficiently characterize vectors expressible as \pproducts\  of single vectors.
For any monoid $(M,+)$ the Expressibility problem then reduces to 
finding a non-negative integer solution of a finite system of linear inequations
over $(M,+)$.
In particular, for monoids $(\Z^d,+)$ this provides an \NP\ algorithm, matching the lower bound
from the feasibility of integer linear programs. 

The second approach (\cref{sec:modules})
reduces Expressibility to the problem of
finding an (not necessarily non-negative) integer solution to a finite linear system.
This assumes that $(M,+)$ is a group and that all base vectors are reversible,
i.e., their inverses are expressible.
We show that this reversibility condition can be verified by checking
the existence of rational solutions of a system over $(M,+)$.
For monoids $(\Z^d,+)$, checking this reversibility condition
and solving the Expressibility problem for reversible instances
is possible in deterministic polynomial time.

We show two applications to the
reachability analysis of counter programs extended with data (in \cref{sec:app}).
The first involves finding state invariants for unordered data Petri nets \cite{lazic08,RSF2011,HLLLST2016}
and the second aplication is the reachability problem for blind counter automata \cite{Grei1978}
extended with data.

\subparagraph*{Related Research and Motivation.}
Our main motivation for studying the Expressibility problem comes from the analysis of
Petri nets extended with data -- the model of \emph{Unordered Petri data nets} (UPDN) of \cite{lazic08},
discussed in \cref{subsec:udpn} -- where data vectors of the form $\vec{v}:\D\to\Z^d$ occur naturally.
We are interested in an invariant sometimes called \emph{state equations} in the Petri net literature.

State equation \cite{SateEquationFirtsPaper} is a fundamental invariant of the reachability relation for Petri nets and one of the important ingredients in the proof of decidability of the reachability
relation \cite{mayr81,kosaraju82,Leroux15}. Some modification of it can be also used as heuristic to improve a performance of the standard backward coverability algorithm for Petri nets
(due to well-structured transition systems \cite{finkel98b,abdulla2000c,SS13,abdulla2011classification}) like it was done in \cite{Qcover}.
The coverability problem has been proven to be \EXPSPACE-complete in \cite{rackoff78,lipton76} and in
\cite{BacwardCoverabilityIsEXPSPACE} it was shown that the backward algorithm matches this complexity.

UDPN were introduced in \cite{lazic08}
as one of the extensions of Petri nets in which the coverability problem remains decidable.
Due to results about undecidability of boundedness for Petri nets with resetting arcs \cite{BoundednesInResetNets},
it is not hard to conclude that UDPN is the only extension of Petri nets, among those proposed in \cite{lazic08}, for which reachability problem may be decidable.
The first indicator that reachability may be decidable for UDPN is a characterization of the coverability set established in the paper \cite{HLLLST2016},
in the same paper, as a conclusion, the place-boundedness problem is proven to be decidable.
The recent development in other classes proposed in \cite{lazic08} can be found in papers \cite{RSF2011,rosavelardo11,rosavelardo14,HLLLST2016},
all those results are focused on better understanding of the coverability relation.

UDPN can also be seen as a restriction of more general \emph{Colored Petri nets} (CPN), see for instance \cite{Coloured_Petri_Nets}.
There is a long history of research in the area of restricted CPNs.
Here, we point to results about invariants and identification of certain syntactic substructures: in \cite{FlowsInRegularCPN,Evangelista2007} authors investigates flows in subclasses CPN.
Another important branch of research concerns structural properties of Algebraic Nets \cite{Petri-Nets-and-Algebraic-Specifications} like detecting
siphons \cite{SiphonsInAlgebraicNets} or other kind of place invariants \cite{HomogeneousEquationsOfAlgebraicPetriNets}.

The third perspective is algebraic methods for Petri nets. Linear algebra and linear programming are one of the most fruitful approaches to Petri nets.
A broad overview of algebraic methods for Petri nets can be found in \cite{Linear-algebraic-and-linear-programming-techniques-for-the-analysis-of-place/transition-net-systems}.
A beautiful application of algebraic techniques are results on reachability in continuous Petri nets \cite{Continuous-Petri-Nets:-Expressive-Power-and-Decidability-Issues}.
One can also find variants of state equation for Petri nets with resetting transitions
\cite{HS14} or with inhibitor arcs \cite{Generalized-State-Equation-for-Petri-Nets}.
Finally, algebraic methods are also used in restricted classes of Petri nets like conflict-free and free-choice
Petri nets \cite{DE95book}.

Finally, it is worth mentioning that UDPN can be interpreted as ordinary Petri nets
for sets with \emph{equality atoms}. We refer the reader to \cite{BKLT13,BKL14}
for work on sets with atoms/nominal sets.
Very similar in spirit to our study of data vectors is the work in \cite{KKOT15} that considers constraint satisfaction problems on infinite structures.

\section{Data Vectors}
\label{sec:notations}
\label{sec:datavectors}

In the sequel $\setD$ is a countable set of elements called \emph{data values}
and $(M,+)$ is a commutative monoid with neutral element $0$.
A \emph{data vector} (also \emph{vector} for short) is a total function $\vec{v}:\setD\rightarrow M$
such that the \emph{support}, the set $\support{\vec{v}}\eqdef \{\alpha\in \D \mid \vec{v}(\alpha)\not=0\}$ is finite.
The monoid operation $+$ is lifted to vectors poinwise, so that
$(\vec{v}+\vec{w})(\alpha) \eqdef \vec{v}(\alpha)+\vec{w}(\alpha)$.

Writing $\circ$ for function composition, we see that $\vec{v}\circ\pi$ is a data vector for any data vector $\vec{v}$
and permutation $\pi:\D\to\D$.
A vector $\vec{x}$ 
is said to be a \emph{\pproduct } of a set $V$ of vectors if there are
$\vec{v_1},\ldots,\vec{v_n}$ in $V$ and permutations
$\theta_1,\ldots,\theta_n$ of $\D$ such that
$$\vec{x}=\sum_{i=1}^{n} \vec{v_i} \circ\theta_{i}.$$

Here, we have to emphasize that it is possible that $\vec{v_i}=\vec{v_j}$ for some $i$ and $j$.

\medskip
When working with vectors of the form $\vec{v}\circ\theta$ for permutations $\theta$,
it will be instrumental to specify $\theta$ indirectly using some injection of $\support{\vec{v}}$ into $\D$.
In the remainder of this section we show that one can always do this.

\medskip
Take any finite subset $\setS$ of $\D$ and $\pi:\setS\to\D$ injective.
Since $\pi^{-1}$ is only a partial function,
define the data vector $\vec{v}\circ\pi^{-1}$
so that any $\beta$ not in the range of $\pi$ maps to $0$:
More precisely, let $(\vec{v} \circ\pi^{-1}):\D\to M$ be defined,
for every $\beta\in\D$ as follows.
$$
\vec{v}\circ\pi^{-1}(\beta)=
  \begin{cases}
      \vec{v}(\alpha) & \text{ if }\pi(\alpha)=\beta\\
      0  & \text{ if } \pi(\alpha)\neq\beta\text{ for all } \alpha\in\D
  \end{cases}
  $$
\begin{lemma}\label{lem:inj2perm}
  Let $\vec{v}$ be a data vector.
  For any permutation $\theta:\D\to\D$ there exists an injection $\pi:\support{\vec{v}}\to\D$,
  such that
  $\vec{v}\circ\pi^{-1} = \vec{v}\circ\theta$,
  and vice versa.
\end{lemma}
\begin{proof}
  Let us introduce $\setS=\support{\vec{v}}$
  and first consider a permutation $\theta$.
  We show that the injection 
  $\pi:\setS\rightarrow\setD$, defined by
  $\pi(\alpha)=\theta^{-1}(\alpha)$ for every $\alpha\in\setS$,
  satisfies $\vec{v}\circ\theta=\vec{v} \circ\pi^{-1}$.
  
  Pick any $\beta\in\setD$ and let us write $\alpha\eqdef\theta(\beta)$.
  If $\alpha\in\setS$ then $\pi(\da)=\db$ so, we have that
  $(\vec{v}\circ\pi^{-1})(\beta)=\vec{v}(\alpha) 
  =(\vec{v}\circ\theta)(\beta)$.
  If $\alpha\not\in\setS$ then
  $(\vec{v}\circ\pi^{-1})(\beta)=0$ and 
$\vec{v}\circ \theta(\db)=\vec{v}(\alpha)=0$, by definition of $\vec{v}\circ\pi^{-1}$
  and because $\setS$ is the support of $\vec{v}$.

  \medskip

  Conversely, let us consider a data injection $\pi$ over $\setS$
  and let us prove that there exists a data permutation $\theta$
  such that $\vec{v}\circ\pi^{-1}=\vec{v}\circ\theta$. We introduce
  $\setT=\pi(\setS)$. Since the restriction of $\pi$ on
  $\setS$ is a bijection onto $\setT$, there exists a
  bijection $\pi':\setT\rightarrow\setS$ denoting its inverse.
  We introduce the sets $X=\setS\backslash\setT$, and
  $Y=\setT\backslash\setS$. Since $\setT$ and $\setS$
  have the same cardinal, it follows that $X$ and $Y$ are two finite
  sets with the same cardinal. Hence, there exists a bijection
  $\pi_{Y,X}':X\rightarrow Y$ and its inverse $\pi_{Y,X}:Y\rightarrow
  X$. We
  introduce the function $\theta$ defined for every $\beta\in\setD$
  as follows:
  $$\theta(\beta)=
  \begin{cases}
    \pi'(\beta) & \text{ if }\beta\in \setT\\
    \pi_{Y,X}'(\beta) & \text{ if }\beta\in X\\
    \beta & \text{ otherwise}
  \end{cases}$$
  Observe that $\theta$ is a bijection since the function
  $\theta'$ defined for every $\alpha\in\setD$ as follows is its inverse:
  $$\theta'(\alpha)=
  \begin{cases}
    \pi(\alpha) &\text{ if }\alpha\in\setS\\
    \pi_{Y,X}(\alpha) & \text{ if }\alpha\in Y\\
    \alpha & \text{ otherwise.}
  \end{cases}
  $$\
  We show that $\vec{v}\circ\pi^{-1}=\vec{v}\circ\theta$.
  Fix some $\beta\in\setD$ and assume first that $\beta\in\setT$.
  There exists
  $\alpha\in\setS$ such that $\pi(\alpha)=\beta$. It follows
  that $\pi'(\beta)=\alpha
  =
  \theta(\beta)$. Hence,
  $(\vec{v}\circ\theta)(\beta)=\vec{v}(\alpha)
  =(\vec{v}\circ\pi^{-1})(\beta)$.
  
  Now, assume that $\beta\not\in \setT$. In that case, notice that
  $\theta(\beta)$ is not in $\setS$ no matter if $\beta\in X$ or
  not. Thus $\vec{v}\circ\theta(\beta)=0$. Observe that if
  $\pi^{-1}(\{\beta\})$ is empty, we deduce that
  $(\vec{v}\circ\pi^{-1})(\beta)=0$. If $\pi^{-1}(\beta)=\{\alpha\}$ then
  $(\vec{v}\circ\pi^{-1})(\beta)=\vec{v}(\alpha)$.
  But since $\beta\not\in\setT$, we deduce
  that $\alpha\not\in\setS$. Therefore $\vec{v}(\alpha)=0$ and
  we derive that $(\vec{v}\circ\pi^{-1})(\beta)=0$. We have proved that
  $\vec{v}\circ\pi^{-1}=\vec{v}\circ\theta$.
\end{proof}

\section{Histograms}\label{sec:Histograms}
\label{sec:hist}
In this section we develop the notion of histograms.
These are combinatorical objects that will be used
in the next section to characterize \pproducts\ over singleton sets $V$.

\begin{definition}\label{def:histogram}
A \emph{histogram} over a finite set $\setS\subseteq \D$
is a total function $H:\setS\times\D\to\N$ such that for some $n\in\setN$, called the \emph{degree} of $H$, the following two conditions hold.
\begin{enumerate}
  \item $\sum_{\beta\in\D}H(\alpha,\beta)=n$ for any $\alpha\in\setS$
  \item $\sum_{\alpha\in\setS}H(\alpha,\beta)\leq n$ for any $\beta\in\D$.
\end{enumerate}
A histogram of degree $n=1$ is called \emph{simple}.
Histograms with the same signature, i.e.~the same sets $\setS$ and $\setD$,
can be partially ordered and summed pointwise and the degree of the sum is the sum of degrees.
The \emph{support} of $H$ is the set $\support{H}\eqdef\{\beta\mid \sum_{\alpha\in\setS} H(\alpha,\beta)>0\}$.
\end{definition}

The following theorem states the main combinatorial property we are interested in,
namely that simple histograms over $\setS$ generate as finite sums the class of all histograms over $\setS$.
In particular, any histogram can be decomposed into finitely many
simple histograms over the same signature (see \cref{fig:hist} for an illustration).

\begin{theorem}\label{thm:histogram}
  A function $H:\setS\times\setD\to \setN$ is a histogram of degree $n\in \N$
  if, and only if, $H$ is the sum of $n$ simple histograms over $\setS$.
\end{theorem}
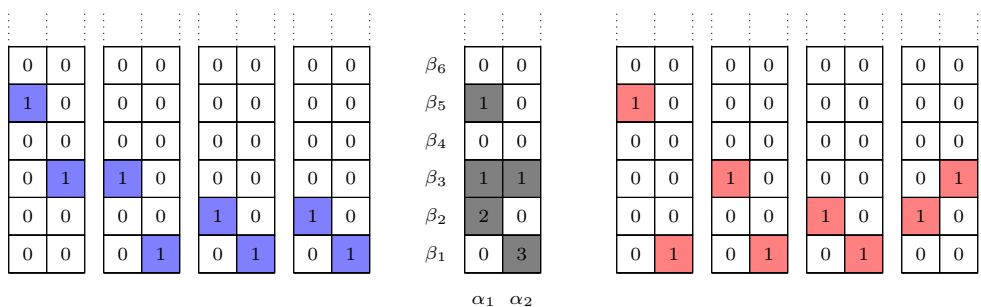
\begin{figure}[H]
\caption{
The center depicts a histogram $H:\setS\x\D\to\N$ of degree 4, where $\setS=\{\alpha_1,\alpha_2\}$
and $\support{H}=\{\beta_1,\beta_2,\beta_3,\beta_5\}$.
To the left (in blue) and to the right (in red)
are decompositions into four simple histograms each.
\label{fig:hist}}
\begin{center}
\vspace{-0.5cm}
 \begin{tikzpicture}[scale=0.5]

 \pgfsetxvec{\pgfpoint{1cm}{0cm}}
 \pgfsetyvec{\pgfpoint{0cm}{1cm}}
 \pgfsetzvec{\pgfpoint{1cm}{0cm}}
 
 \coordinate(O) at (0,0,0);
 \coordinate(X) at (1,0,0);
 \coordinate(Y) at (0,1,0);
 \coordinate(Z) at (0,0,1);
 \def\xmax{2}
 \def\ymax{6}
  \newcommand{\hist}[1]{
\foreach \x in {0,...,\xmax}
{
    \draw[-] ($(#1)+(\x,0,0)$)--($(#1)+(\x,\ymax,0)$);
    \draw[dotted] ($(#1)+(\x,\ymax,0)$) -- ($(#1)+(\x,\ymax+1,0)$);
    \foreach \y in {0,...,\ymax}
    {
        \draw[-] ($(#1)+(0,\y,0)$)--($(#1)+(\xmax,\y,0)$);
    }
    
}
  }
  \newcommand{\prostokat}[3][]{
    \draw[#1] ($(#2)$)--($(#2)+(X)$)--($(#2)+(X)+(Y)$)--($(#2)+(Y)$)--cycle;
    \node[font=\scriptsize] at ($(#2)+0.5*(X)+0.5*(Y)$){#3};
  }

  \begin{scope}[]
      \coordinate (offset) at (0,0,0);
      \def\col{black!50}
      \prostokat[draw=black,fill=white]{$(offset)+(0,0,0)$}{0}
      \prostokat[draw=black,fill=white,fill=\col]{$(offset)+(1,0,0)$}{3}
      \prostokat[draw=black,fill=white,fill=\col]{$(offset)+(0,1,0)$}{2}
      \prostokat[draw=black,fill=white]{$(offset)+(1,1,0)$}{0}
      \prostokat[draw=black,fill=white,fill=\col]{$(offset)+(0,2,0)$}{1}
      \prostokat[draw=black,fill=white,fill=\col]{$(offset)+(1,2,0)$}{1}
      \prostokat[draw=black,fill=white]{$(offset)+(0,3,0)$}{0}
      \prostokat[draw=black,fill=white]{$(offset)+(1,3,0)$}{0}
      \prostokat[draw=black,fill=white,fill=\col]{$(offset)+(0,4,0)$}{1}
      \prostokat[draw=black,fill=white]{$(offset)+(1,4,0)$}{0}
      \prostokat[draw=black,fill=white]{$(offset)+(0,5,0)$}{0}
      \prostokat[draw=black,fill=white]{$(offset)+(1,5,0)$}{0}

      \foreach \i in {1,...,\xmax}{
      \prostokat[draw=white]{$(offset)+(\i-1,-1.25,0)$}{$\alpha_\i$}{fill=white}
      }
      \foreach \i in {1,...,\ymax}{
      \prostokat[draw=white]{$(offset)+(0,\i-1,0)+(-1.25,0,0)$}{$\beta_{\i}$}{fill=white}
      }
      \hist{offset}
  \end{scope}
      
\begin{scope}[]
  \def\col{blue!50}
  \def\zdist{-2.5}
  \def\zfirstdist{-4.5}
  \begin{scope}[]
      \coordinate (offset) at (0,0,0*\zdist+\zfirstdist);
      \prostokat[draw=black,fill=white]{$(offset)+(0,0,0)$}{0}
      \prostokat[draw=black,fill=white,fill=\col]{$(offset)+(1,0,0)$}{1}
      \prostokat[draw=black,fill=white,fill=\col]{$(offset)+(0,1,0)$}{1}
      \prostokat[draw=black,fill=white]{$(offset)+(1,1,0)$}{0}
      \prostokat[draw=black,fill=white]{$(offset)+(0,2,0)$}{0}
      \prostokat[draw=black,fill=white]{$(offset)+(1,2,0)$}{0}
      \prostokat[draw=black,fill=white]{$(offset)+(1,3,0)$}{0}
      \prostokat[draw=black,fill=white]{$(offset)+(0,3,0)$}{0}
      \prostokat[draw=black,fill=white]{$(offset)+(1,4,0)$}{0}
      \prostokat[draw=black,fill=white]{$(offset)+(0,4,0)$}{0}
      \prostokat[draw=black,fill=white]{$(offset)+(0,5,0)$}{0}
      \prostokat[draw=black,fill=white]{$(offset)+(1,5,0)$}{0}
      \hist{offset}
  \end{scope}
  \begin{scope}[]
      \coordinate (offset) at (0,0,1*\zdist+\zfirstdist);
      \prostokat[draw=black,fill=white]{$(offset)+(0,0,0)$}{0}
      \prostokat[draw=black,fill=white,fill=\col]{$(offset)+(1,0,0)$}{1}
      \prostokat[draw=black,fill=white,fill=\col]{$(offset)+(0,1,0)$}{1}
      \prostokat[draw=black,fill=white]{$(offset)+(1,1,0)$}{0}
      \prostokat[draw=black,fill=white]{$(offset)+(0,2,0)$}{0}
      \prostokat[draw=black,fill=white]{$(offset)+(1,2,0)$}{0}
      \prostokat[draw=black,fill=white]{$(offset)+(1,3,0)$}{0}
      \prostokat[draw=black,fill=white]{$(offset)+(0,3,0)$}{0}
      \prostokat[draw=black,fill=white]{$(offset)+(1,4,0)$}{0}
      \prostokat[draw=black,fill=white]{$(offset)+(0,4,0)$}{0}
      \prostokat[draw=black,fill=white]{$(offset)+(0,5,0)$}{0}
      \prostokat[draw=black,fill=white]{$(offset)+(1,5,0)$}{0}
      \hist{offset}
  \end{scope}
  \begin{scope}[]
      \coordinate (offset) at (0,0,2*\zdist+\zfirstdist);
      \prostokat[draw=black,fill=white]{$(offset)+(0,0,0)$}{0}
      \prostokat[draw=black,fill=white,fill=\col]{$(offset)+(1,0,0)$}{1}
      \prostokat[draw=black,fill=white]{$(offset)+(0,1,0)$}{0}
      \prostokat[draw=black,fill=white]{$(offset)+(1,1,0)$}{0}
      \prostokat[draw=black,fill=white,fill=\col]{$(offset)+(0,2,0)$}{1}
      \prostokat[draw=black,fill=white]{$(offset)+(1,2,0)$}{0}
      \prostokat[draw=black,fill=white]{$(offset)+(1,3,0)$}{0}
      \prostokat[draw=black,fill=white]{$(offset)+(0,3,0)$}{0}
      \prostokat[draw=black,fill=white]{$(offset)+(1,4,0)$}{0}
      \prostokat[draw=black,fill=white]{$(offset)+(0,4,0)$}{0}
      \prostokat[draw=black,fill=white]{$(offset)+(0,5,0)$}{0}
      \prostokat[draw=black,fill=white]{$(offset)+(1,5,0)$}{0}
      \hist{offset}
  \end{scope}
  \begin{scope}[]
      \coordinate (offset) at (0,0,3*\zdist+\zfirstdist);
      \prostokat[draw=black,fill=white]{$(offset)+(0,0,0)$}{0}
      \prostokat[draw=black,fill=white]{$(offset)+(0,1,0)$}{0}
      \prostokat[draw=black,fill=white]{$(offset)+(1,0,0)$}{0}
      \prostokat[draw=black,fill=white]{$(offset)+(1,1,0)$}{0}
      \prostokat[draw=black,fill=white]{$(offset)+(0,2,0)$}{0}
      \prostokat[draw=black,fill=white,fill=\col]{$(offset)+(1,2,0)$}{1}
      \prostokat[draw=black,fill=white]{$(offset)+(0,3,0)$}{0}
      \prostokat[draw=black,fill=white]{$(offset)+(1,3,0)$}{0}
      \prostokat[draw=black,fill=white,fill=\col]{$(offset)+(0,4,0)$}{1}
      \prostokat[draw=black,fill=white]{$(offset)+(1,4,0)$}{0}
      \prostokat[draw=black,fill=white]{$(offset)+(0,5,0)$}{0}
      \prostokat[draw=black,fill=white]{$(offset)+(1,5,0)$}{0}
      \hist{offset}
  \end{scope}
\end{scope}
\begin{scope}[]
  \def\col{red!50}
  \def\zdist{2.5}
  \def\zfirstdist{4}
  \begin{scope}[]
      \coordinate (offset) at (0,0,0*\zdist+\zfirstdist);
      \prostokat[draw=black,fill=white]{$(offset)+(0,0,0)$}{0}
      \prostokat[draw=black,fill=white,fill=\col]{$(offset)+(1,0,0)$}{1}
      \prostokat[draw=black,fill=white]{$(offset)+(0,1,0)$}{0}
      \prostokat[draw=black,fill=white]{$(offset)+(1,1,0)$}{0}
      \prostokat[draw=black,fill=white]{$(offset)+(0,2,0)$}{0}
      \prostokat[draw=black,fill=white]{$(offset)+(1,2,0)$}{0}
      \prostokat[draw=black,fill=white]{$(offset)+(0,3,0)$}{0}
      \prostokat[draw=black,fill=white]{$(offset)+(1,3,0)$}{0}
      \prostokat[draw=black,fill=white,fill=\col]{$(offset)+(0,4,0)$}{1}
      \prostokat[draw=black,fill=white]{$(offset)+(1,4,0)$}{0}
      \prostokat[draw=black,fill=white]{$(offset)+(0,5,0)$}{0}
      \prostokat[draw=black,fill=white]{$(offset)+(1,5,0)$}{0}
      \hist{offset}
  \end{scope}
  \begin{scope}[]
      \coordinate (offset) at (0,0,1*\zdist+\zfirstdist);
      \prostokat[draw=black,fill=white]{$(offset)+(0,0,0)$}{0}
      \prostokat[draw=black,fill=white,fill=\col]{$(offset)+(1,0,0)$}{1}
      \prostokat[draw=black,fill=white]{$(offset)+(0,1,0)$}{0}
      \prostokat[draw=black,fill=white]{$(offset)+(1,1,0)$}{0}
      \prostokat[draw=black,fill=white,fill=\col]{$(offset)+(0,2,0)$}{1}
      \prostokat[draw=black,fill=white]{$(offset)+(1,2,0)$}{0}
      \prostokat[draw=black,fill=white]{$(offset)+(1,3,0)$}{0}
      \prostokat[draw=black,fill=white]{$(offset)+(0,3,0)$}{0}
      \prostokat[draw=black,fill=white]{$(offset)+(1,4,0)$}{0}
      \prostokat[draw=black,fill=white]{$(offset)+(0,4,0)$}{0}
      \prostokat[draw=black,fill=white]{$(offset)+(0,5,0)$}{0}
      \prostokat[draw=black,fill=white]{$(offset)+(1,5,0)$}{0}
      \hist{offset}
  \end{scope}
  \begin{scope}[]
      \coordinate (offset) at (0,0,2*\zdist+\zfirstdist);
      \prostokat[draw=black,fill=white]{$(offset)+(0,0,0)$}{0}
      \prostokat[draw=black,fill=white,fill=\col]{$(offset)+(1,0,0)$}{1}
      \prostokat[draw=black,fill=white,fill=\col]{$(offset)+(0,1,0)$}{1}
      \prostokat[draw=black,fill=white]{$(offset)+(1,1,0)$}{0}
      \prostokat[draw=black,fill=white]{$(offset)+(0,2,0)$}{0}
      \prostokat[draw=black,fill=white]{$(offset)+(1,2,0)$}{0}
      \prostokat[draw=black,fill=white]{$(offset)+(1,3,0)$}{0}
      \prostokat[draw=black,fill=white]{$(offset)+(0,3,0)$}{0}
      \prostokat[draw=black,fill=white]{$(offset)+(1,4,0)$}{0}
      \prostokat[draw=black,fill=white]{$(offset)+(0,4,0)$}{0}
      \prostokat[draw=black,fill=white]{$(offset)+(0,5,0)$}{0}
      \prostokat[draw=black,fill=white]{$(offset)+(1,5,0)$}{0}
      \hist{offset}
  \end{scope}
  \begin{scope}[]
      \coordinate (offset) at (0,0,3*\zdist+\zfirstdist);
      \prostokat[draw=black,fill=white]{$(offset)+(0,0,0)$}{0}
      \prostokat[draw=black,fill=white]{$(offset)+(1,0,0)$}{0}
      \prostokat[draw=black,fill=white,fill=\col]{$(offset)+(0,1,0)$}{1}
      \prostokat[draw=black,fill=white]{$(offset)+(1,1,0)$}{0}
      \prostokat[draw=black,fill=white]{$(offset)+(0,2,0)$}{0}
      \prostokat[draw=black,fill=white,fill=\col]{$(offset)+(1,2,0)$}{1}
      \prostokat[draw=black,fill=white]{$(offset)+(0,3,0)$}{0}
      \prostokat[draw=black,fill=white]{$(offset)+(1,3,0)$}{0}
      \prostokat[draw=black,fill=white]{$(offset)+(0,4,0)$}{0}
      \prostokat[draw=black,fill=white]{$(offset)+(1,4,0)$}{0}
      \prostokat[draw=black,fill=white]{$(offset)+(0,5,0)$}{0}
      \prostokat[draw=black,fill=white]{$(offset)+(1,5,0)$}{0}
      \hist{offset}
  \end{scope}
\end{scope}
\end{tikzpicture}

\vspace{-0.5cm}
\end{center}
\end{figure}
For the proof of \cref{thm:histogram} we need a lemma from graph theory.
We refer the reader to \cite{book:Deistel} for relevant definitions
and recall here only that in a graph $(V,E)$,
a \emph{matching} of a set $S\subseteq V$ of nodes is a set $M\subseteq E$
of pairwise non-adjacent edges that covers all nodes in $S$.

\begin{lemma}\label{lem:twoMatchings}
Let $G=(L\cup R, E)$ be a bipartite graph. If there is a matching of $L'\subseteq L$ and a matching of $R'\subseteq R$
  then there is a matching of $L'\cup R'$. 
\end{lemma}
\begin{proof}
  Suppose, $M_L$ and $M_R$ are matchings that matches $L'$ and $R'$, respectively.
  Let $G'\eqdef (L\cup R, M_L \cup M_R)$ be a subgraph of $G$. We
  construct a matching $M$ of $L'\cup R'$ as a matching in $G'$.
  Observe that $G'$ is a union of single nodes, paths and cycles and $M$ can be constructed in every strongly connected component independently.
  
  We claim that for any strongly connected component $C$ we can find a matching witch matches all elements in $C\cap L'$ and $C\cap R'$.
  This, if proved, ends the proof of Lemma~\ref{lem:twoMatchings}.
  
  First of all, every node in $L'\cup R'$ has a degree at least $1$ so the claim holds for single nodes immediately.
  
  If the strongly connected component is a cycle (in bipartite graph) or a path of an even length then there is a perfect matching in it so the claim holds as well.
  
  The case of paths of odd length is the most complicated one.
  Without loosing of generality, suppose that the first node $x$ is in $L'\cup R'$. Indeed, if it is not then we match all nodes except
  $x$ and this case is done. 
  
  Without loss of generality, we can assume that $x\in L'$, as $x\in R'$ is symmetric.
  We prove that the path has to end in vertex from the set $L\setminus L'$.
  Indeed path has to end in $L$ as it's length is odd. Furthermore, 
  consider a walk along the path starting from $x$; to every element of $L'\setminus \{x\}$ on the path $C$ we enter via an edge from $M_R$, but then we can leave it via 
  an edge from $M_L$ as from any vertex in $L'$ there is outgoing edge in $M_L$. Thus, the path can not end in the element from $L'$ and thus the last
  vertex has to belong to $L\setminus L'$. Now we match all vertices except of the last one.
\end{proof}

\begin{proof}[Proof of \cref{thm:histogram}]
If $H=\sum_{j=1}^nH_j$, where the $H_j$ are simple histograms over $\setS$,
then from $\sum_{\beta\in\setD}H(\alpha,\beta)=\sum_{j=1}^n\sum_{\beta\in\setD}H_j(\alpha,\beta)$
and because $\sum_{\beta\in\setD}H_j(\alpha,\beta)=1$ 
we derive that
$\sum_{\beta\in\setD}H(\alpha,\beta)= n$ for every
$\alpha\in\D$. In the same way the second histogram condition for $H$ follows
from those of the $H_j$. So $H$ is a histogram over $\setS$.

\smallskip
For the converse direction, the proof proceeds by induction on $n$,
the degree of the histogram $H$.
If $n=1$ then $H$ is simple and the claim trivially holds.
Suppose now that the claim holds for histograms of degree $n$ 
and consider a histogram $H:\setS\times\setD\rightarrow\setN$ of degree $n+1$.
We show how that there exists a simple histogram $X:\setS\x\D\to\N$ such that 

\begin{enumerate}
    \item $X(\alpha,\beta)\le H(\alpha,\beta)$ for every $\alpha\in
      \setS$, and every $\beta\in \setD$, and
    \item for every $\db\in\setD$ if $\sum_{\alpha\in\setS}H(\alpha,\beta)=n+1$ then
    $\sum_{\alpha\in\setS}X(\alpha,\beta) =1$.
\end{enumerate}

The first condition then guarantees that the function $Y:\setS\x\D\to\N$ with $Y=H-X$ is well defined
and satisfies
\begin{equation*}
\sum_{\beta\in \setD} Y(\alpha,\beta)=\left(\sum_{\beta\in \setD} H(\alpha,\beta)\right)-\left(\sum_{\beta\in \setD} X(\alpha,\beta)\right)=(n+1)-1=n \ \text{ for all }\alpha\in \setS. 
\end{equation*}
The second condition implies that $\sum_{\alpha\in \setS} Y(\alpha,\beta)\leq n$ for all $\beta\in\setD$.
Hence, $Y$ is a histogram of degree $n$ and the claim follows by induction hypothesis.

\medskip
To show the existence of a suitable simple histogram $X$, we 
consider now the bipartite graph $G$ where the sets of nodes are $\setS$
and $\setB\eqdef\support{H}$
and where there is an edge between $\alpha$ and $\beta$ whenever $H(\alpha,\beta)>0$
(We assume here w.l.o.g.~that $\setS$ and $\support{H}$ are disjoint; otherwise take $\setB$ as
some suitable duplication).
Moreover, let $\setT$ denote the set of those ``maximal'' data values $\beta$
where $\sum_{\alpha\in\D}H(\alpha,\beta)=n+1$.
Note that $\setT\subseteq\setB$.

We claim that the required simple histogram $X$ exists iff
there is a matching in the graph $G$ that matches both $\setS$ and $\setT$.
Indeed, any such histogram $X$ provides a matching
$M\eqdef\{(\alpha,\beta)\mid X(\alpha,\beta)=1\}$.
By the first histogram condition, $M$ matches all nodes in ${\setS}$;
and all nodes $\beta\in {\setT}$ are matched since $X$ 
must satisfy $\sum_{\alpha\in \setS}X(\alpha,\beta)=1$.
Conversely, for a given matching $M$ we define 
$X(\alpha,\beta) = 1$ if $(\alpha,\beta)\in M$ and $0$ otherwise.
It is easy to check that $X$ is a simple histogram that satisfies required properties.

\medskip
To finish the proof we show that a matching $M$ of $\setS\cup\setT$ exists. 
By \cref{lem:twoMatchings}, it suffices to find two matchings, one of $\setS$
and one of $\setT$.
In both cases, we will make use of Hall's marriage Theorem~\cite{book:Deistel}.
Writing $\neighbours{\setS}$ for the \emph{neighbourhood}
of a set $\setS$ of nodes (the set of nodes $v\notin \setS$ adjacent to some node in $\setS$),
this theorem states that
in a finite bipartite graph there is a matching of a set $\setS$ of nodes
if, and only if every subset $\setS'\subseteq \setS$
has at least as many neighbours as elements:
\begin{equation}
    \label{eq:hall}
    \card{\setS'} \le \card{\neighbours{\setS'}}.
\end{equation}

We start by proving the existence of a matching of $\setS$.
If we label the edges $(\alpha,\beta)$ in our graph by the respective values $H(\alpha,\beta)$,
then we observe that
the total weight of edges connecting any subset $\setS'\subseteq \setS$
is at most the total weight of edges connecting its neighbours:
$\sum_{\alpha\in \setS'}\sum_{\beta\in\setD} H(\alpha,\beta)
\le
\sum_{\alpha\in \setS}\sum_{\beta\in\neighbours{\setS'}} H(\alpha,\beta)$.
Consequently,
  
\begin{align*}
    \card{\setS'}\cdot (n+1)=\sum_{\substack{\alpha\in \setS'\\\beta\in \setD}} H(\alpha,\beta)
   &\leq \sum_{\substack{\alpha\in\setS\\\beta\in \neighbours{\setS'}}} H(\alpha,\beta)
   \le \sum_{\beta\in \neighbours{\setS'}} (n+1)
   = \card{\neighbours{\setS'}}\cdot(n+1).
  \end{align*}
  The first equality is due to the first histogram condition and
  the second inequality is by the second histogram condition.
  So all subsets $\setS'\subseteq\setS$ satisfy \cref{eq:hall},
  so Hall's theorem applies and there exists a matching of $\setS$.

  The proof that a matching of $\setT$ exists follows the same pattern.
  For any subset $\setT'\subseteq \setT$ we get
  $$\card{\setT'}\cdot (n+1)=
  \sum_{\beta\in \setT'} \sum_{\alpha\in \setS} H(\alpha,\beta)
  \leq
  \sum_{\alpha\in \neighbours{\setT'}} \sum_{\beta\in \setD} H(\alpha,\beta)=
  \card{\neighbours{\setT'}}\cdot (n+1),
  $$
  where the first equality holds by the definition of $\setT$.
  So $\setT$ satisfies the assumption of Hall's theorem and we conclude
  that some matching of $\setT$ exists, as required.
\end{proof}

\section{Expressibility}\label{sec:naturals}
In this section, we show that histograms provide a natural tool for
deciding if a data vector is a \pproduct\  of others.
We first establish the connection between histograms and \pproducts,
and then (in \cref{thm:boundedsupport}) that \pproducts\  can be represented by
histograms with bounded support sets.
Finally, we derive an \NP\ complexity upper bound for the Expressibility problem
for vectors over monoids $(\Z^d,+)$.
\medskip

\noindent
Given a data vector $\vec{v}$ and a histogram $H$ over $\setS\eqdef \support{\vec{v}}$ we define the vector $\eval{\vec{v},H}$ by
$$
 \eval{\vec{v},H}(\db)\eqdef
\sum_{\da\in \setS}\vec{v}(\da)H(\da,\db).
$$
Observe that $\mathit{eval}$ is a homomorphism in the sense that
for any vector $\vec{v}$ and histograms $H_1,H_2$ over $\setS$ it holds that
\begin{equation}
    \label{eq:eval-hom}
\eval{\vec{v},H_1+H_2}=\eval{\vec{v},H_1}+\eval{\vec{v},H_2}.
\end{equation}

\noindent
Recall that by \cref{lem:inj2perm}, \pproducts\  are of the form
$\vec{x}=\sum_{i=1}^n\vec{v}_i\circ\pi_i^{-1}$ where $\pi_i:\support{\vec{v}_i}\to\D$
are injections.
We now associate each injective function $\pi:\setS\to\D$ with the
simple histogram $H_\pi$ over $\setS$ defined as
$$H_\pi(\alpha,\beta)\eqdef
\begin{cases}
  1 & \text{ if }\beta=\pi(\alpha)\\
  0 & \text{ otherwise}
\end{cases}
$$
Notice that conversely, each simple histogram $H:\setS\x\D\to\D$
provides a unique injection $\pi_H:\setS\to\D$ satisfying $H(\alpha,\pi_H(\alpha))=1$ for every $\alpha\in \setS$.
So, $H_{\pi_H}=H$.
The next lemma makes the connection between histograms and \pproducts\ .

\begin{lemma}
    \label{lem:inj-hist}
    \label{lem:perm-hist}
    Let $\vec{v}$ be a vector 
    and $\pi:\support{\vec{v}}\to\D$ injective.
    Then
    $\vec{v}\circ\pi^{-1} = \eval{\vec{v},H_\pi}$.
\end{lemma}
\begin{proof}
If $\pi(\alpha)=\beta$ then
$\vec{v}\circ\pi^{-1}(\beta) = \vec{v}(\alpha)
=\sum_{\alpha'\in\setS}\vec{v}(\alpha')H(\alpha',\beta) = \eval{\vec{v},H}(\beta)
$
where the second equation holds because $H$ is simple.
If $\pi(\alpha)\neq\beta$ for all $\alpha$ then
$\vec{v}\circ\pi^{-1}(\beta) = 0 =\sum_{\alpha\in\setS}\vec{v}(\alpha)\cdot 0
=\eval{\vec{v},H}(\beta).
$
\end{proof}

\begin{lemma}\label{lem:single}
    Let $\vec{v}$ be a vector. A vector $\vec{x}$ is a \pproduct\  of
  $\{\vec{v}\}$ if, and only if, there exists a histogram $H$ over $\support{v}$
  such that
  $\vec{x}=\eval{\vec{v},H}$.
\end{lemma}
\begin{proof}
  Assume first that $\vec{x}$ is a \pproduct\  of
  $\{\vec{v}\}$. Then there are
  permutations
  $\theta_1,\ldots,\theta_n$
  such that
  $\vec{x}=\sum_{j=1}^n\vec{v}\circ\theta_j$.
  By \cref{lem:inj2perm},
  $\vec{x}=\sum_{j=1}^n\vec{v}\circ\pi_j^{-1}$
  for injections  $\pi_j:\support{\vec{v}}\to\D$.
  From \cref{lem:perm-hist} we derive $\vec{x}=\sum_{j=1}^n\eval{\vec{v},H_{\pi_j}}$
  and by \cref{eq:eval-hom}
  the histogram $H\eqdef\sum_{j=1}^nH_{\pi_j}$
  satisfies
  $\vec{x}=\eval{\vec{v},H}$.

  \smallskip
  Conversely, let us assume that there exists an histogram $H$ over
  $\support{v}$ such that $\vec{x}=\eval{\vec{v},H}$.
  \cref{thm:histogram} shows that
  $H$ can be decomposed as $H=\sum_{j=1}^nH_j$ where $H_j$ are
  simple histograms over $\support{\vec{v}}$. From \cref{eq:eval-hom}
  it follows that $\vec{x}=\sum_{j=1}^n\eval{\vec{v},H_j}$,
  and since all $H_j$ are simple,
  there are injections $\pi_1,\ldots,\pi_n$
  with $H_j=H_{\pi_j}$.
  The claim thus follows by \cref{lem:inj-hist} and \cref{lem:inj2perm}.
\end{proof}

We now show how to bound the supports of histograms $H_{\vec{v}}$ such
that $\vec{x}=\sum_{\vec{v}\in V}\eval{\vec{v},H}$ with respect to
$\vec{x}$ and $V$.

  \begin{theorem}\label{thm:boundedsupport}
    If $\vec{x}$ is a \pproduct\  of $V$
    then
    $\vec{x}=\sum_{\vec{v}\in V}\eval{\vec{v},H_{\vec{v}}}$
    where for each vector $\vec{v}\in V$, $H_\vec{v}$ is a histogram
    over $\support{\vec{v}}$, and
    $\card{\bigcup_{\vec{v}\in V}\support{H_{\vec{v}}}}$
    is bounded by $\card{\support{\vec{x}}}+1+\sum_{\vec{v}\in V}(2\card{\support{\vec{v}}}-1)$.
  \end{theorem}
  \begin{proof}
    Any \pproduct\  
    of $V$
    is a sum
    $\vec{x}=\sum_{\vec{v}\in V}\vec{x_v}$, where each $\vec{x_v}$ is a \pproduct\  of $\{\vec{v}\}$.
    So the existence of histograms $H_\vec{v}$ with $\vec{x_v}=\eval{\vec{v},H_\vec{v}}$ is guaranteed by \cref{lem:single}.

    For each histogram $H_\vec{v}$ we write $n_\vec{v}$ for its degree
    and define the set of \emph{big} data values as
    $$\setB_{\vec{v}}\eqdef
    \left\{\db\in\setD \mid
    \sum_{\da\in\support{\vec{v}}}H_{\vec{v}}(\da,\db)> \frac{n_{\vec{v}}}{2}\right\}.$$

    \noindent
    We can estimate the cardinality of $\setB_\vec{v}$ by
    $|\setB_{\vec{v}}|\leq 2|\support{\vec{v}}|-1$.
    Indeed, if $\setB_{\vec{v}}$ is empty, the property is immediate. Otherwise, we have 
    $\sum_{\db\in \setB_{\vec{v}}}\sum_{\da\in
      \support{\vec{v}}}H_{\vec{v}}(\da,\db)> |\setB_{\vec{v}}|\frac{n_{\vec{v}}}{2}$.
    We also have $\sum_{\db\in \setB_{\vec{v}}}\sum_{\da\in
      \support{\vec{v}}}H_{\vec{v}}(\da,\db)\leq \sum_{\da\in
      \support{\vec{v}}}\sum_{\db\in \setD}H_{\vec{v}}(\da,\db)=
     \card{\support{\vec{v}}} \cdot n_{\vec{v}}$.
    Therefore, $|\setB_{\vec{v}}|\leq 2|\support{\vec{v}}|-1$ holds.

    \medskip
    To provide the bound claimed in the theorem,
    suppose that the histograms $H_\vec{v}$ are chosen such that
    their combined support $\setT\eqdef\bigcup_{\vec{v}\in V}\support{H_\vec{v}}$ is minimal.
    We show that this set cannot have more than
    $\card{\support{\vec{x}}}+1+\sum_{\vec{v}\in V}(2\card{\support{\vec{v}}}-1)$
    elements, which implies the claim.

    Suppose towards a contradiction that $\card{\setT}$ exceeds this bound.
    Then it must contain two distinct elements
    $\db_1$ and $\db_2$
    that are both not in
    $\bigcup_{\vec{v}\in V}\setB_{\vec{v}}$
    nor in $\support{\vec{x}}$.
    Notice that the first condition, that $\beta_1,\beta_2$ are not big in any histogram
    $H_\vec{v}$ guarantees that
    \begin{equation}
        \label{eq:notbig}
    \sum_{\alpha\in \support{\vec{v}}}H_\vec{v}(\alpha,\beta_1)
    + \sum_{\alpha\in \support{\vec{v}}}H_\vec{v}(\alpha,\beta_2) \le n_\vec{v} \ \ \ \ \   \text{ for all } \vec{v}\in V.
    \end{equation}
    Based on $\beta_1$ and $\beta_2$ we introduce, for each $\vec{v}\in V$,
    the function
    $F_{\vec{v}}:\support{\vec{v}}\x\D\to\N$ as
    $$F_{\vec{v}}(\da,\db)=
    \begin{cases}
      H_{\vec{v}}(\da,\db_1)+H_{\vec{v}}(\da,\db_2) &\text{ if
      }\db=\db_1\\
      0 & \text{ if }\db=\db_2\\
      H_{\vec{v}}(\da,\db) & \text{ otherwise.}
    \end{cases}$$
    Then for any $\alpha\in\support{v}$ we have
    $\sum_{\beta\in \D}F_\vec{v}(\alpha,\beta) = n_\vec{v}$ and moreover, by
    \cref{eq:notbig}, we have
    $\sum_{\alpha\in \support{\vec{v}}}F_\vec{v}(\alpha,\beta_1) \le n_\vec{v}$.
    So $F_\vec{v}$ is a histogram over $\support{\vec{v}}$ of degree $n_\vec{v}$.
    We claim that
    $$\sum_{\vec{v}\in V}   \eval{\vec{v},F_{\vec{v}}}=\sum_{\vec{v}\in V}\eval{\vec{v},H_{\vec{v}}}.$$ 
    Indeed, $F_{\vec{v}}$ trivially satisfies
    $\eval{\vec{v},H_{\vec{v}}}=\eval{\vec{v},F_{\vec{v}}}$ for all data except of $\db_1$ and $\db_2$.
    Thus 
    $$\vec{x}(\db)=\left(\sum_{\vec{v}\in V} \vec{x_{\vec{v}}}\right)(\db)=\left(\sum_{\vec{v}\in V} \eval{\vec{v},F_{\vec{v}}}\right)(\db)\text{ for all }\db\not\in\{ \db_1,\db_2\}.$$
    Moreover, $\db_2\not\in \support{\vec{x}}$ so $\vec{x}(\db_2)=0$. On the other hand $\left(\sum_{\vec{v}\in V} \eval{\vec{v},F_{\vec{v}}}\right)(\db_2)=0$ as 
    for every $\vec{v}$ it holds that $\eval{\vec{v},F_{\vec{v}}}(\db_2)=0$.
    Finally, $\db_1\not\in \support{\vec{x}}$ so $\vec{x}(\db_1)=0$. On the other hand 
    \begin{equation*}
    \begin{split}
    \left(\sum_{\vec{v}\in V} \eval{\vec{v},F_{\vec{v}}}\right)(\db_1)=\left(\sum_{\vec{v}\in V} \eval{\vec{v},H_{\vec{v}}}\right)(\db_1)+
    \left(\sum_{\vec{v}\in V} \eval{\vec{v},H_{\vec{v}}}\right)(\db_2)\\
    =\left(\sum_{\vec{v}\in V} \vec{x_{\vec{v}}}\right)(\db_1)+\left(\sum_{\vec{v}\in V} \vec{x_{\vec{v}}}\right)(\db_2)=\vec{x}(\db_1)+\vec{x}(\db_2)=0+0=0.
    \end{split}
    \end{equation*}

   We conclude that $\vec{x}=\sum_{\vec{v}\in V}\eval{\vec{v},F_{\vec{v}}}$.
   But this contradicts
   the minimality of $\card{\setT}$ as it strictly includes
   $\bigcup_{\vec{v}\in V}\support{F_{\vec{v}}} = \setT\setminus\{\beta_2\}$.
  \end{proof}

\begin{corollary}
    \label{cor:expr-in-NP}
    The Expressibility problem for data vectors with values in $(\Z^d,+)$ is \NP-complete.
\end{corollary}
\begin{proof}
    We show the upper bound only, as a matching lower bound holds already for singleton domains $\D$,
    where the problem is equivalent to the feasibility of integer linear programs.

    By \cref{thm:boundedsupport}, positive instances imply the existence of
    histograms $H_{\vec{v}}$ over $\support{\vec{v}}$, with polynomially bounded
    support, 
    with $\vec{x}=\sum_{\vec{v}\in V}\eval{\vec{v},H_\vec{v}}$.
    By \cref{lem:single}, the existence of such histograms is also a sufficient condition
    for $\vec{x}$ to be a \pproduct\  of $V$.
    
    Due to the bound from \cref{thm:boundedsupport}, the histogram conditions
    as well as the condition that $\vec{x}=\sum_{\vec{v\in V}}\eval{\vec{v},H_\vec{v}}$
    can be expressed as a system of linear constraints with polynomially
    many inequalities and unknowns,
    which has a non-negative integer solution iff these conditions are satisfied.
    The claim thus follows from standard results for integer linear programming.
\end{proof}

\section{Reversibility}\label{sec:modules}
In this section we consider data vectors $\vec{v}:\setD\rightarrow G$ where
$(G,+)$ is a commutative group.
In this case the set of all vectors is itself a commutative group with identity $\vec{0}$.
We write $-\vec{v}$ for the inverse of vector $\vec{v}$ and $\vec{v}-\vec{w}\eqdef \vec{v}+(-\vec{w})$.

\begin{definition}
    \label{def:reversible}
    A vector $\vec{x}:\D\to G$ is \emph{reversible in} $V$ if both $\vec{x}$ and $-\vec{x}$ are \pproducts\  of $V$.
    A set of vectors $V$ is reversible if every vector $\vec{v}\in V$ is reversible in $V$.
 \end{definition}

We will provide in this section a way to reduce the Expressibility problem
to the membership problems in finitely generated subgroups of $(G,+)$,
assuming the given set of data vectors is reversible.
This result stated as \cref{thm:modules}.
We also show (as \cref{thm:reversibility}),
that checking the reversibility condition amounts to solving a finite linear system over $(G,+)$.

\medskip
Our constructions are based on the homomorphism \emph{weight}, which projects data vectors into
the underlying group: the weight of a vector $\vec{v}:\D\to G$ is the element of $G$ defined as
    $$\weight{\vec{x}}\eqdef\sum_{\da\in\setD}\vec{x}(\da).$$
    For a set $V$ of data vectors define $\weight{V}\eqdef\{\weight{\vec{v}} \mid \vec{v}\in V\}$.

In addition we introduce some useful notation. 
Fix any total order on $\setD$.
The \emph{rotation} of a finite set $\setS\subseteq \D$ 
is the permutation
$\cycshift{\setS}:\D\to\D$
defined as 
$$
    \cycshift{\setS}(\alpha) \eqdef
    \begin{cases}
        \min\{\setS\}, &\mbox{ if } \alpha=\max\{\setS\}\\
        \min\{\beta\in \setS\mid \beta >\alpha\}, &\mbox{ if } \max\{\setS\}\neq\alpha\in \setS\\
        \alpha, &\mbox{ if } \alpha\notin \setS.
    \end{cases}
$$

This allows to express for instance the vector $\vec{v}\circ\cycshift{\{\alpha,\beta\}}$,
which results from the vector $\vec{v}$ by exchanging the values of $\alpha$ and
$\beta$.
Clearly, for any finite set $\setS\subseteq\D$, 
the $\card{\setS}$-fold composition of
$\cycshift{\setS}$ with itself is the identity on $\D$. 

Finally, we introduce vectors $\RS{\setS}{v}:\D\to\Z^d$ as
$$
    \RS{\setS}{v} \eqdef \sum_{i=0}^{|\setS|-1} \vec{v} \circ \cycshift{\setS}^i
$$
where $\vec{v}$ is a data vector, $\support{\vec{v}}\subseteq \setS\subseteq\D$ is finite, and the superscripts denote $i$-fold iteration.
It is the result of summing up all different $\setS$-rotations of $\vec{v}$.
This vector is useful because it is a \pproduct\  of $\vec{v}$
that ``equalizes'' all values for $\alpha\in \setS$ to $\weight{\vec{v}}$, as stated in the
proposition below.

\begin{proposition}
    \label{lem:equalize}
    Let $\vec{v}:\D\to G$ be a data vector
    and $\setS\subseteq \D$ finite such that $\support{\vec{v}}\subseteq \setS$.
    \begin{enumerate}
        \item $\RS{\setS}{v}$ is a \pproduct\  of $\{\vec{v}\}$.
        \item 
            $\RS{\setS}{v}(\alpha) = \weight{\vec{v}}$ if $\alpha\in \setS$
            and
            $\RS{\setS}{v}(\alpha) = 0$ if $\alpha\notin \setS$.
    \end{enumerate}
\end{proposition}

\subsubsection*{Identifying Reversible Sets of Vectors}
\begin{restatable}{theorem}{lemreversability}
    \label{lem:reversability}
    \label{thm:reversibility}

Let $V$ be a set of data vectors 
and $\vec{x}\in V$.
Then $\vec{x}$ is reversible in $V$ if, and only if,
$\weight{\vec{x}}$ is reversible in $\weight{V}$,
i.e.,
there exist $\vec{v}_1,\vec{v}_2,\dots,\vec{v}_n\in V$ such that
$-\weight{\vec{x}} = \sum_{i=1}^n \weight{\vec{v}_i}$.
\end{restatable}
\begin{proof}[Proof of \cref{thm:reversibility}]
    For the only if direction we need to show that
    if $-\vec{x} = \sum_{i=1}^n\vec{v_i}\circ\theta_i$ for vectors
    $\vec{v_i}\in V$ and permutations $\theta_i:\D\to\D$.
    Since $\mathit{weight}$ is a homeomorphism we observe that
    $-\weight{\vec{x}}$ is expressible as a sum $\sum_{i=1}^j\weight{\vec{w_i}}$, where $\vec{w_i}\in V$.
    The claim follows from the fact that
    $\weight{\vec{v_i}\circ\theta_i}=\weight{\vec{v_i}}$
    for all $\vec{v_i}:\D\to G$.
  
  For the opposite direction
  assume vectors $\vec{v_1},\vec{v_2},\ldots,\vec{v_n}\in V$ such that
      $
     -\weight{\vec{x}} =\sum_{j=1}^{n} \weight{\vec{v_j}}
     $
     and let $\setS\eqdef  \bigcup_{i=1}^n\support{\vec{v_i}}$.
  First, we aim to show that
  \begin{equation}
      \label{eq:rev1}
      -\RS{\setS}{\vec{x}} 
      = \sum_{j=1}^{n}\RS{\setS}{\vec{v_j}}, 
  \end{equation}
  that is, $-\RS{\setS}{\vec{x}}(\alpha)= \sum_{j=1}^{n}\RS{\setS}{\vec{v_j}}(\alpha)$ for all $\alpha\in\D$.
  As $\vec{x}\in V$, by definition of $\RS{\setS}{\vec{x}}$, this trivially holds for $\alpha \notin \setS$.
  For the remaining $\alpha\in \setS$, note that by point~2 of \cref{lem:equalize}
  we have $\RS{\setS}{\vec{v}}(\alpha) = \weight{\vec{v}}$
  for any $\vec{v}\in V$. In particular this holds for $\vec{x}$ and all $\vec{v_j}$. So, 
  \begin{equation}
      -\RS{\setS}{\vec{x}}(\alpha)
      = -\weight{\vec{x}}
      = \sum_{j=1}^{n}{\weight{\vec{v_j}}}
      = \sum_{j=1}^{n}\RS{\setS}{\vec{v_j}}(\alpha)
  \end{equation}
  which proves \cref{eq:rev1}.
  Unfolding the definition of $\RS{\setS}{\vec{x}}$ we therefore see that
  \begin{align*}
      & - \RS{\setS}{\vec{x}}
      =
     -\left(\vec{x} +  \sum_{i=1}^{\card{\setS}-1} \automorph{\cycshift{\setS}^i}{\vec{x}}\right)
     =
    \sum_{j=1}^{n} \RS{\setS}{\vec{v_j}}
  \end{align*}
  and consequently that
  $-\vec{x} = (\sum_{i=1}^{\card{\setS}-1} \automorph{\cycshift{\setS}^i}{\vec{x}})
    +\sum_{j=1}^{n} \RS{\setS}{\vec{v_j}}
  $.
  Now,
  $(\sum_{i=1}^{\card{\setS}-1} \automorph{\cycshift{\setS}^i}{\vec{x}})$ is clearly a \pproduct\  of $\{\vec{x}\}$ and thus also of $V$.
  By point~1 of \cref{lem:equalize}, also $\sum_{j=1}^{n}\RS{\setS}{\vec{v_j}}$
  is a \pproduct\  of $V$.
  We conclude that $-\vec{x}$ is a \pproduct\  of $V$ and therefore that $\vec{x}$ is reversible in $V$.
  \end{proof}
 
\begin{corollary}\label{cor:polyTimeProcedureToTestRevers}
    Let
    $d\in\N$ and
    $V$ a finite set of vectors $\vec{v}:\D\to \Z^d$.
    There is a polynomial time procedure that checks if $V$ is reversible.
\end{corollary}
\begin{proof}
    By \cref{lem:reversability}, it suffices to verify for all $\vec{v}\in V$
    that $-\weight{\vec{v}}$ is the sum of elements in ${\weight{V}}$.
    In other words, we need to check for an element $h$ of a finite
    set $H=\{h_1,\ldots,h_k\} \subseteq\setZ^d$, that there exist
    $n_1,\ldots,n_k\in\setN$ with
        $-h=n_1h_1+\cdots+n_kh_k$.

    We show that
    the condition above is satisfied
    if, and only if,
    there exists $\lambda_1,\ldots,\lambda_k\in\Q_{\geq 0}$ such that
    $-h=\lambda_1h_1+\cdots+\lambda_kh_k$.
    The ``only if'' direction is immediate.
    For the converse, assume
    factors $\lambda_1,\ldots,\lambda_k\in\Q_{\geq 0}$
    such
    that $-h=\lambda_1h_1+\cdots+\lambda_kh_k$.
    There exists a positive integer $p$ such that $p\lambda_j\in\setN$ for
    every $j$ and thus
    $-h=(p\lambda_1)h_1+\cdots+(p\lambda_k)h_k+(p-1)h$.
    We have proved the claim.
    
    Now, checking if $-h=\lambda_1h_1+\cdots+\lambda_kh_k$
    has a solution in the non-negative rationals
    is doable in polynomial time
    using linear programming \cite{LProgPTIME}.
  \end{proof}

\subsubsection*{Solving Expressibility in Reversible Sets of Vectors}
In the remainder of this subsection we proof \cref{thm:modules}. We start with a lemma.
All constructions in this section assume $\max_{\vec{v}\in V}|\support{\vec{v}}|<|\setD|$,
that there exists at least one fresh datum in the domain.

\newcommand{\lift}[2]{{\llbracket#2\mapsto #1\rrbracket}}
We first prove that some special data vectors are \pproducts\ of $V$. 
Those vectors are defined by introducing for
every element $g\in G$ and every data value $\da\in\setD$, the data
vector $\lift{g}{\da}$ defined for every $\db\in\setD$
by:
$$\lift{g}{\da}(\db)\eqdef\begin{cases}
  g & \text{ if }\db=\da\\
  0 & \text{ otherwise}
\end{cases}
$$

\begin{lemma}\label{lem:weight}
  Let $V$ be a finite set of data vectors such that
  $\bigcup_{\vec{v}\in V}\support{\vec{v}}\subsetneq\setD$.
  Then, $\lift{g}{\db}$ is a \pproduct\  of $V$
  for every $g$ in the subgroup of $(G,+)$ generated by  $\weight{V}$, and for every $\db\in\setD$.
\end{lemma}
\begin{proof}
  Since $g$ is in the subgroup generated by $\weight{V}$, there
  exist a sequences $\vec{v_1},\ldots,\vec{v_n}$
  of elements in $V$ such that:
  $$g=\sum_{j=1}^n\weight{\vec{v_j}}$$
  Consider now the vector $\vec{x}\eqdef\sum_{j=1}^n\vec{v_j}$.
  It has three relevant properties:
  \begin{enumerate}
      \item it is a \pproduct\  of $V$,
      \item its support is contained in $\bigcup_{\vec{v}\in V}\support{\vec{v}}$, and
      \item it satisfies $g=\weight{\vec{x}}$.
  \end{enumerate}
  By point~2 and the assumption that the combined support $\bigcup_{\vec{v}\in V}\support{\vec{v}}$
  is strictly included in $\D$, we can pick some $\da\not\in\support{\vec{x}}$.
  Let $\setS,\setT\subseteq\D$ be defined as
  $$
  \setS\eqdef\support{\vec{x}} \qquad\text{and}\qquad
  \setT\eqdef{\support{x}\cup\{\alpha\}}
  $$
  Then by \cref{lem:equalize} point~1,
  both $\RS{\setS}{\vec{x}}$ and
  and $\RS{\setT}{\vec{x}}$ are \pproducts\  of $V$.
  Since $V$ is reversible, so is the inverse $-\RS{\setS}{\vec{x}}$.
  It remains to observe that
  $$
  \lift{g}{\alpha} = \RS{\setT}{\vec{x}} -\RS{\setS}{\vec{x}}.
  $$
  Indeed, for all $\delta\notin\setT$ we have
  $\lift{g}{\alpha}(\delta) = \RS{\setS}{\vec{x}}(\delta)=\RS{\setT}{\vec{x}}(\delta)=0$.
  For all $\delta\in\setS$,
  by \cref{lem:equalize} point~2, it holds that
  $\RS{\setS}{\vec{x}}(\delta)=\RS{\setT}{\vec{x}}(\dd)=\weight{\vec{x}}=g$ and therefore that
  $\lift{g}{\alpha}(\delta) = \RS{\setS}{\vec{x}}(\delta)-\RS{\setT}{\vec{x}}(\delta) = g-g=0$.
  For the last case that $\delta=\alpha\in\setT\setminus\setS$, again by
  \cref{lem:equalize} point~2, we have
  $\lift{g}{\alpha}(\delta) = \RS{\setS}{\vec{x}}(\delta)-\RS{\setT}{\vec{x}}(\delta) = g-0=g$.
  Now, the vector $\lift{g}{\db}=\lift{g}{\da}\circ \cycshift{\{\da,\db\}}$ which completes the proof
\end{proof}

\begin{lemma}\label{lem:place}
  Let $V$ be a finite, reversible set of data vectors such that
  $\bigcup_{\vec{v}\in V}\support{\vec{v}}\subsetneq\setD$.
     The data vector $\lift{g}{\da}-\lift{g}{\db}$ is a \pproduct
     of $V$ for every $g$ in the subgroup of $(G,+)$ generated by
   $\{\vec{v}(\delta) \mid \delta\in\setD,~\vec{v}\in V\}$, and for every $\da,\db\in\setD$.
\end{lemma}
\begin{proof}
  Let $\setT\eqdef\bigcup_{\vec{v}\in V}\support{\vec{v}}$ and
  assume w.l.o.g.~that $\da\not=\db$ since otherwise the claim is trivial.
  It suffices to show the claim for $\da\in\setT$ and $\db\not\in\setT$.

  We first show that there exists a data vector $\vec{x}$ that is a
  \pproduct\  of $V$ such that $\vec{x}(\da)=g$ and such that
  $\support{\vec{x}}\subseteq\setT$.
  Since $g$ is in the subgroup generated by
  $\{\vec{v}(\delta) \mid \delta\in\setD,~\vec{v}\in V\}$,
  there exist vectors $\vec{v_1},\ldots,\vec{v_n}\in V$
  and data values
  $\delta_1,\ldots,\delta_n\in \D$
  such that:
  $$g=\sum_{j=1}^n\vec{v_j}(\delta_j).$$
  We can assume without loss of generality that $\vec{v_j}(\delta_j)$ are not
  equal to zero and hence $\delta_j\in\setT$.

  Let $\theta_j\eqdef\cycshift{\{\da,\delta_j\}}:\D\to\D$ be the permutation that exchanges $\da$ and
  $\delta_j$
  and consider the vector $\vec{x}$, defined as follows.
  $$\vec{x}=\sum_{j=1}^n\vec{v_j}\circ\theta_j$$
  Observe that $\vec{x}$ is a \pproduct\  of $V$ and
  $\vec{x}(\da)=g$ as it was required. Moreover, since
  $\delta_j,\da\in\setT$, we deduce that
  $\support{\vec{v_j}\circ\theta_j}$
  is included in $\setT$ and therefore that $\support{\vec{x}}\subseteq \setT$.

  \medskip
  To show the claim,
  let ${\theta}\eqdef\cycshift{\{\alpha,\beta\}}$ be the permutation that swaps $\da$ and $\db$
  and consider the vector
  $$\vec{y}\eqdef \vec{x}-\vec{x}\circ{\theta}.$$
  For all $\delta\in\setD\setminus\{\da,\db\}$ we get
  $\vec{y}(\delta)=\vec{x}(\delta)-\vec{x}({\theta}(\delta))=\vec{x}(\delta)-\vec{x}(\delta)=0.$
  Moreover, $\vec{y}(\da)=\vec{x}(\da)-\vec{x}({\theta}(\da))=g-\vec{x}(\db)=g$, similarly $\vec{y}(\db)=-g$.
  We concluding that $\vec{y}$ is a \pproduct\  of $V$ and $\vec{y}=\lift{g}{\da}-\lift{g}{\db}$.
\end{proof}

We can now prove our main theorem.
\begin{theorem}
  \label{thm:modules}
  Let $V$ be a finite, reversible set of data vectors with
  $\bigcup_{\vec{v}\in V}\support{\vec{v}}\subsetneq\setD$.
  A data vector $\vec{x}$ 
  is a \pproduct\  of $V$ if, and only if, the
  following two conditions hold.
  \begin{itemize}
  \item $\weight{\vec{x}}$ is in the subgroup of $(G,+)$ generated by $\{\weight{\vec{v}} \mid
    \vec{v}\in V\}$, and
  \item $\vec{x}(\da)$ is in the subgroup of $(G,+)$ generated by $\{\vec{v}(\delta) \mid
    \delta\in \setD,~\vec{v}\in V\}$ for every $\da\in\setD$.
    \end{itemize}
\end{theorem}
\begin{proof}
  If $\vec{x}$ is a \pproduct\  of $V$
  there exists a sequence $\vec{v_1},\ldots,\vec{v_n}$ of data vectors in $V$ and
  a sequence $\theta_1,\ldots,\theta_n$ of data permutations such
  that $\vec{x}=\sum_{j=1}^n\vec{v_j}\circ\theta_j$. We derive that
  $\weight{\vec{x}}=\sum_{j=1}^n\weight{\vec{v_j}\circ\theta_j}$. Since
  $\weight{\vec{v_j}\circ\theta_j}=\weight{\vec{v_j}}$, it follows that $\weight{\vec{x}}$ is
  in the group generated by $\weight{V}$.
  Moreover, for every $\da\in\setD$, we have
    $\vec{x}(\da)=\sum_{j=1}^n\vec{v_j}(\theta_j(\da))$. Thus
    $\vec{x}(\da)$ is in the group generated by $\{\vec{v}(\delta) \mid
    \delta\in \setD,~\vec{v}\in V\}$.
  
  For the converse direction, assume that $\vec{x}$ is a data vector satisfying the two
  conditions. We pick $\delta\in\setD$. From condition~2 and \cref{lem:place} we derive that
  for every $\da\in\support{\vec{x}}$, the data vector
  $\lift{\vec{x}(\da)}{\da}-\lift{\vec{x}(\da)}{\delta}$ is a \pproduct\ 
  of $V$. It follows that the $\vec{y}\eqdef\sum_{\da\in\support{\vec{x}}}\left(\lift{\vec{x}(\da)}{\da}-\lift{\vec{x}(\da)}{\delta}\right)$ is
  a \pproduct\  of $V$. Notice that this vector is
  equal to $\vec{x}-\lift{\weight{\vec{x}}}{\delta}$.
  By condition~1, \cref{lem:weight} applies and implies that
  $\lift{\weight{\vec{x}}}{\delta}$ is a \pproduct\  of $V$.
  So, $\vec{x}$ must be a \pproduct\  of $V$.
\end{proof}
\begin{corollary}\label{cor:modules}
    Let $\D$ be infinite, $d\in\N$, and
    $V$ a finite, reversible set of vectors $\vec{v}:\D\to \Z^d$.
    There is a polynomial time procedure that
    determines if a given target vector $\vec{x}:\D\to\Z^d$ is a \pproduct\  of $V$.
\end{corollary}

\begin{remark}
  To motivate the freshness assumption, $\bigcup_{\vec{v}\in V}\support{\vec{v}} \subseteq \setD$
  consider the following example, which shows that the two conditions in the claim of \cref{thm:modules}
  are not necessarily sufficient on its own.
  
  $\setD$ is the finite set $\{\da_1,\ldots,\da_k\}$ where
  $\da_1,\ldots,\da_k$ are distinct and $k\geq 2$. Assume that there exists
  $m\in G$ such that $k\cdot m\not=0$. We introduce $V=\{\vec{v}, -\vec{v}\}$ where
  $\vec{v}=\lift{m}{\da_1}+\cdots+\lift{m}{\da_k}$ and $\vec{x}=\lift{(k\cdot m)}{\da_1}$. Observe
  that $\vec{x}$ satisfies the two conditions of \cref{thm:modules}. Assume
  by contradiction that $\vec{x}$ is a \pproduct\  of
  $V$. Since $\vec{v}\circ\theta=\vec{v}$ for every data
  permutation $\theta$ it follows that $\vec{x}=z\cdot \vec{v}$ for some
  $z\in\setZ$. Thus $0=\vec{x}(\da_2)=z\cdot \vec{v}(\da_2)=z\cdot m$, and
  $k\cdot m=\vec{x}(\da_1)=z\cdot \vec{v}(\da_1)=z\cdot m$. Hence $k\cdot m=0$ and we get a
  contradiction. Hence $\vec{x}$ is not a \pproduct\  of
  $V$.
\end{remark}

\section{Applications}\label{sec:applications}
\label{sec:app}
\subsection{Unordered data Petri nets}\label{subsec:udpn}
Unordered data nets extend the classical model of Petri nets by
allowing each token to carry a datum from a countable set $\setD$.
We recall the definition from \cite{RSF2011,rosavelardo11}.
\newcommand{\MS}[1]{{#1}^\oplus}
A multiset over some set $X$ is a function $M:X\to\N$.
The set $\MS{X}$ of all multisets over $X$ is ordered pointwise, and the multiset union of $M,M'\in\MS{X}$
is $(M\oplus M')\in\MS{X}$ with $(M\oplus M')(\alpha)\eqdef M(\alpha)+M'(\alpha)$ for all $\alpha\in X$.
If $M\ge M'$, then the multiset difference $(M\ominus M')$ is defined as the unique $X\in\MS{X}$
with $M=M'\oplus X$.

\newcommand{\Var}{\mathit{Var}}
\begin{definition}
An unordered Petri data net (UPDN) over domain $\setD$ is a tuple $(P,T,F)$ where
$P$ is a finite set of \emph{places},
$T$ is a finite set of \emph{transitions} disjoint from $P$,
and ${F:(P\x T)\cup(T\x P) \to \MS{\Var}}$ is a \emph{flow} function that assigns
each place $p\in P$ and transition $t\in T$
a multiset of over \emph{variables} in $\Var$.

\smallskip
A \emph{marking} is a function $M:P\to \MS{\setD}$.
Intuitively, $M(p)(\da)$ denotes the number of tokens of type $\da$ in place $p$.
A transition $t$ is \emph{enabled} in marking $M$ with \emph{mode} $\sigma$ if
$\sigma:\Var\to D$ is an injection such that
$\sigma(F(p,t)) \le M(p)$ for all $p\in P$.
There is a step $M\step{}M'$ between markings $M$ and $M'$ if
there exists $t$ and $\sigma$ such that $t$ is enabled in $M$ with mode $\sigma$,
and for all $p\in P$,
$$M'(p)=M(p) \ominus \sigma(F(p,t)) \oplus \sigma(F(t,p)).$$
The transitive and reflexive closure of $\step{}$ is written as $\step{*}$.
\end{definition}
Notice that UDPN are a generalization of ordinary P/T nets, which have only one type of token, i.e.~$\setD=\{\bullet\}$.

The decidability status of the reachability problem for UDPN,
which asks if $M\step{*}M'$ holds for given markings $M,M'$ in a given UDPN,
is currently open.
We will discuss here a necessary condition for positive instances,
an invariant sometimes called \emph{state equations} in the Petri net literature.

First, notice that markings in UDPN can be seen as data vectors over the monoid $(\Z^d,+)$.
For any marking $M$ and place $p$, $M(p)$ is a multiset over $\setD$, i.e.~a data vector $M(p):\setD\to\N$.
Markings are therefore isomorphic to data vectors $M:\setD\to\N^d$, where $d=\card{P}$.

Similarly, flow function provides multisets $F(p,t)$ and $F(t,p)$ over the variables $\Var$,
so we can associate to each transition $t$ the corresponding data vectors
$F(\bullet,t)$ and ${F(t,\bullet):\Var\to\N^d}$,
defined as $F(\bullet,t)(x)\eqdef [F(p_1,t)(x),F(p_2,t)(x)\ldots F(p_{\card{P}},t)(x)]$
and analogously, $F(t,\bullet)(x)\eqdef [F(t,p_1)(x),F(t,p_2)(x)\ldots F(t,p_{\card{P}})(x)]$.
Further, the \emph{deplacement} $\Delta(t)$ of transition $t$ is $\Delta(t)\eqdef F(t,\bullet) - F(\bullet,t)$,
which is a function $\Delta(t):\Var\to\Z^d$ over $(\Z^d,+)$.
Now, $M\step{}M'$ iff there is a transition $t$ and injection $\sigma:\Var\to\setD$
such that
$M - F(\bullet,t)\circ\sigma^{-1} \ge \vec{0}$
and $M' = M+\Delta(t)\circ\sigma^{-1}$.
A necessary condition for reachability can thus be formulated as follows.

\begin{proposition}
    If $M\step{*}M'$ then there exists a sequence $t_1,t_2,\ldots,t_k\in T$ of transitions 
    and a sequence $\sigma_1,\sigma_2,\ldots,\sigma_k$ of injections such that
    $M'-M = \sum_{i+1}^k\Delta(t_i)\circ\sigma^{-1}$.
\end{proposition}

We say markings $M$ and $M'$ satisfy the state-equation, if there are transitions and injections satisfy the condition above. A direct consequence of \cref{cor:expr-in-NP} is that one can check this condition in \NP.
\begin{theorem}
    There is an \NP\ algorithm that checks if 
    for any two markings of a UPDN satisfy the state equation.
\end{theorem}
The complexity of checking state equations for UDPN thus matches that of the same problem for
ordinary Petri nets (via linear programming).

\medskip
For UDPN where the set of transition effects is itself reversible,
Expressibility can even be decided in polynomial time.
\begin{theorem}
    Let $(P,T,F)$ be a UDPN such that $\{\Delta(t) \mid t\in T'\}$ is reversible.
    Checking if two markings satisfy the state equation is in \P.
\end{theorem}
This directly follows from \cref{cor:modules}.
Notice that this reversibility condition on the transition effects is
a fairly natural condion. For instance, if a UDPN is reversible in the usual
Petri net parlance, i.e., if its reachability relation $\step{*}$ is symmetric,
then the set $\{\Delta(t) \mid t\in T\}$ is reversible
in the sense of \cref{def:reversible}.
Indeed, otherwise there must be some $t'\in T$ where $-\Delta(t')$ is not a \pproduct\  of
$\{\Delta(t)\mid t\in T\}$. So there are markings $M,M'$ and injection $\sigma$
with $M'=M+\Delta(t')\circ\sigma^{-1}$ and $M'\not\step{*}M$. 

Finally, we remark that by \cref{cor:polyTimeProcedureToTestRevers},
we can in polynomial time check if a given UDPN satisfies the
reversibility condition.

\subsection{Blind Counter Automata}
Blind counter automata \cite{Grei1978}
are finite automata equipped with a number of registers that store integer values
and which can be independently incremented or decremented in each step.
These systems correspond to vector addition systems with states (VASS)
over the integers \cite{HS14}, where transitions are always enabled.
The model can be equipped with data in a natural way.

\begin{definition}
 An \emph{unordered data blind counter automaton} is given by a finite labelled transition system $\Aut \eqdef(Q, E, L)$ where edges in $E$ are labelled by the function $L$ 
 with data vectors in $\funs{\Z^k}{\D}$.
 
 A \emph{configuration} of the automaton is a pair $(q,\vec{v})$ where $q\in Q$ is a state and $\vec{v}\in \funs{\Z^k}{\D}$. 
There is a step $(q,\vec{f})\step{e}(q',\vec{g})$ between two configurations $(q,\vec{f}),(q',\vec{g})$
if $e=(q,q')\in E$ and $\vec{g}=\vec{f}+L((q,q'))\circ\theta$ for a permutation $\theta:\D\to\D$. 
The reachability relation is a transitive closure of the step relation and we denote it by $\step{*}$.
Finally, a sequence of configurations and transitions of a form $(q_0,\vec{x_0})\step{e_1}(q_1,\vec{x_1})\step{e_2}\ldots \step{e_n}(q_n,\vec{x_n})$ we call a \emph{path}.
\end{definition}

\begin{theorem}
The reachability problem for an unordered data blind counters automaton can be solved in $\NP$.
\end{theorem}
\begin{proof}{(\it Sketch).}
Any path $\Pi$ from some initial configuration $(q_0, \vec{x_0})$ to some final configuration $(q_f, \vec{x_f})$ 
can be described as a skeleton path $S$ of length at most $\card{Q}^2$ and a multiset $C$ of simple cycles connected to it (like in \cite{ReachabilityInTwoVASS,SkeletonPathInOCA}).
The key properties of the skeleton path are:
\begin{itemize}
 \item it uses transitions from the path $\Pi$,
 \item it visits every state that is visited by $\Pi$,
 \item it starts in $(q_0, \vec{x_0})$ and ends in $(q_f,\vec{x})$ where $\vec{x}$ is some data vector,
 \item its length is bounded by $\card{Q}^2$.
\end{itemize}
The fact that multiset of edges can be decomposed into a set of cycles is equivalent to the condition that for every state the number of 
incoming edges is equal to the number of outgoing edges in $C$. 
Having above we first guess a skeleton path $S$ and next we solve a system of linear inequalities 
which binds the usage of different simple cycles with the desired effect of them.

Precisely, the algorithm first introduces $\card{Q}$ additional counters to our automaton and 
for every edge $(p,q)$ we change label of it to $L_{new}(p,q)\defeq(\vec{v}+\vec{v}_{p,q,\dc})$ where $\dc$ is a fresh data value and $\vec{v}_{p,q,\dc}$ 
is a data vector such that:
\begin{align*}
    \vec{v}_{p,q,\dc}(d)(\da)&=\begin{cases}
      1 & \text{ if } d=q\text{ and } \da=\dc\\
      -1 & \text{ if } d=p\text{ and } \da=\dc\\
      0 & \text{ otherwise.}
    \end{cases}
\end{align*}

Now we guess the skeleton path, and in addition an instantiations of labels taken along the edges of the skeleton path.  
Suppose that the skeleton path is $q_0\step{e_1}q_1\step{e_2}\ldots \step{e_n}q_f$ and
the instantiations labels looks as follows $L(e_1)\circ\theta_1,L(e_2)\circ\theta_2,\ldots, L(e_n)\circ\theta_n$.
Let $\setS$ denote a set of states visited by the skeleton path.
Now, what remains it to calculate the number of occurrence of edges taken in the cyclic part of our path or in other words we need to express
$$\vec{x}-\vec{x_0}-\left(\sum_{i=1}^n L(e_i)\circ\theta_i\right)$$
as a sum $\sum_{j=1} \vec{v_j}\circ\theta_j$ where $\vec{v_j}$ are labels of edges starting in $S$.
This is an instance of the expressibility problem for data vectors over $(\Z^d,+)$,
which is solvable in \NP\ by \cref{cor:expr-in-NP}.
To conclude, positive instances of the reachability problem are witnessed by
a skeleton path, and a solution for the resulting instance of the expressibility problem,
where the base vectors are labels of edges starting in states visited by the skeleton path.
\end{proof}

\newpage
\bibliography{conferences,journalsabbr,aautocleaned}

\end{document}